\documentclass[prl,reprint]{revtex4-1}
\usepackage{braket}
\input{Qcircuit}

\usepackage{mathtools}

\newtheorem{lem}{Lemma}
\newtheorem{thm}[lem]{Theorem}

\newtheorem{cor}[lem]{Corollary}
\newtheorem{ex}[lem]{Example}

\newenvironment{proof}[1][Proof]{\noindent\textbf{#1.} }{\ \rule{0.5em}{0.5em}}

\usepackage{amsmath, amssymb} 

\DeclareMathOperator{\CINC}{CINC}
\DeclareMathOperator{\INC}{INC}
\DeclareMathOperator{\diag}{diag}
\DeclareMathOperator{\poly}{poly}
\DeclareMathOperator{\IQP}{IQP}

\begin{document}

\begin{abstract} 
We generalize an efficient decomposition method for phase-sparse diagonal operators by Welch \textit{et al.}\ to qudit systems.
The phase-context aware method focusses on cascaded entanglers whose decomposition into multi-controlled $\INC$-gates can be optimized by the choice of a proper signed base-$d$ representation for the natural numbers.


While the gate count of the best known decomposition method for general diagonal operators on qubit systems scales with $\mathcal{O}(2^n)$, 
 the circuits synthesized by the Welch algorithm for diagonal operators with $k$ distinct phases are upper-bounded by $\mathcal{O}(n^2k)$, which is generalized to $\mathcal{O}(dn^2k)$ for the qudit case in this paper.
\end{abstract}

\title{Phase context decomposition of diagonal unitaries for higher-dimensional systems} 
\author{Kerstin Beer}
\author{Friederike Anna Dziemba}
\affiliation{Institut f\"ur Theoretische Physik, Leibniz Universit\"at Hannover, Germany}
\date{April 26, 2016}
\maketitle

\section{Introduction}

Diagonal unitary operators form a very restricted class amongst all unitary quantum gates. Despite this 
it has been proven that an efficient quantum circuit consisting of diagonal gates in the conjugate input basis cannot be efficiently classically simulated unless the polynomial hierarchy collapses at the third level \cite{IQP}. Hence the complexity class $\IQP$ (instantaneous quantum polynomial time) formed by these quantum circuits has attracted attention during the last years \cite{yoshiMurao2014, ni, fujii, matsuo}.

Indeed diagonal operators play a central role in many quantum algorithms, for example in
\begin{itemize}
\item the oracle query in Grover's algorithm \cite[6.1.2]{nielsen},
\item quantum optimization \cite{welch2},
\item the simulation of quantum dynamics \cite{zalka, kassal} and
\item decoupling -- the important primitive of quantum Shannon theory can be achieved by random diagonal unitaries \cite{decoupling1, decoupling2}.
\end{itemize}

Implementations of quantum algorithms motivates the study of decomposition methods for diagonal unitaries. Moreover such a decomposition raises interest as a subroutine within the compilation of an arbitrary quantum operation $V$ based on the
\begin{itemize}
\item spectral decomposition \cite{bullock, brennen},
\item Givens' QR method \cite{QR}: $V=QR$ with $Q$ Givens rotations and $R$ diagonal,
\item real quantum computation \cite{realSynthesis}; $V=O_1 D O_2$ with $O_1,O_2$ orthogonal and $D$ diagonal,
\item Cosine-Sine decomposition \cite{cosineSineDecomp, bullockOptimal}: $V=(U_1 \oplus U_2) W (U_3 \oplus U_4)$ with $D:= (SH \otimes \mathbb{I})W(HS^\dagger \otimes \mathbb{I})$ diagonal,
\item conditioned computation where the circuits for the conditional operations are already known up to a relative phase.
\end{itemize}

While circuits decomposing arbitrary unitaries scale with $\mathcal{O}(n^2 2^{2n})$ \cite[4.5.1 -- 4.5.2]{nielsen}, the best known compiling algorithm for diagonal unitaries provides circuits of size $\mathcal{O}(2^n)$ \cite{bullockOptimal}. The exponential growth is avoided in the setting of phase-sparse unitaries with $k$ distinct phases. This setting is studied by the decomposition method of \cite{welch}, which we generalize in this paper from qubit to qudit systems resulting in a circuit scaling of $\mathcal{O}(dn^2k)$.

Qudit systems and their advantages have been studied in \cite{diwei, greentree, lanyon}, while extensive work on the synthesis of qudit operations has been done by \cite{brennen, rosenbaum, bullock, diwei, logic}. 
Extending compiling methods to qudit systems is significant since many implementation architectures exhibit a natural qudit form.

Most of the previously mentioned algorithms containing diagonal unitary operators can be easily adapted for qudit systems. For example, higher dimensional generalizations of Grover's algorithm have been studied in \cite{groverQudits1, groverQudits2, groverQudits3, groverQudits4}. 
The aim of quantum optimization -- finding the ground state of the Hamiltonian $H=-\sum_x g(x) \ket{x}\bra{x}$ to determine the maximum of the function $g(x)$ -- requires the implementation of $e^{-i H t} = \sum_x e^{it g(x)}\ket{x}\bra{x}$, which does not favor any specific underlying dimension structure.
Using a binary encoding of the discrete grid for the originally continuous wavefunction in the simulation of quantum dynamics \cite{zalka, kassal} is not physically motivated but rather arbitrary as well. There is no obstacle in considering a qudit encoding of the grid. Only the above mentioned decoupling method is proven specifically for diagonal unitaries on qubits. But since the underlying realization of decoupling by approximate unitary 2-designs has been proven for arbitrary dimensional systems \cite{decouplingProof} it would actually be interesting to study if a sufficiently approximate 2-design can also be realized by diagonal unitaries in higher dimensional systems.

All these examples motivate to consider diagonal unitaries on general qudit structures and study suitable compiling algorithms for them.


\section{Overview of the algorithm}

Throughout this paper we use the variable $d$ for the dimension of a single qudit and the variable $n$ for the number of qudits the diagonal operator acts on.

The presented algorithm considers the \emph{phase context} of a diagonal operator by splitting it into gate blocks for each of its distinct phases. 
Each block is built from a single-qudit phase gate and two so-called \emph{cascaded entanglers} which can be decomposed into single-qudit multiplication and addition operations and $\land_1$-\ and $\land_2$-gates. The latter are defined as singly- and doubly-controlled $\INC$-operations, where $\INC$ is a single-qudit gate with $\INC\ket{t}=\ket{t+1}$ (operations on the labels of basis states are considered $\bmod{d}$ throughout this paper).

It is known that already $\land_1$, together with all single-qudit gates, forms a universal gate set for higher dimensional quantum computation \cite{brylinski}. Thus one could apply a method like \cite{brennen} to even further decompose the $\land_2$-operations, which can be regarded as higher dimensional generalization of the Toffoli-gate. But the consideration of $\land_2$-gates as elementary allows us to end the decomposition of the cascaded entanglers at a point that reveals their basic structure and which is exact even if the single-qudit operations should be restricted to some finite gate set. Experimental realizations of $\land_2$ have been proposed \cite{diwei}.

The decomposition of the cascaded entanglers into controlled $\INC$-gates is based on a number-theoretical approach. At this point the authors of \cite{welch} focus on the binary representation of integers and mention signed binary expansion as an alternative method. In this paper the central theorem is directly formulated for any signed base-$d$ expansion. In many cases this allows further reduction of the number of required $\land_1$- and $\land_2$-gates.

\section{Phase context decomposition of diagonal unitaries}

We generalize the method of \cite{welch} for the decomposition of diagonal unitary operators to qudit systems.
This method takes into consideration the \textit{phase context} of an operator, i.e. a diagonal unitary
	\begin{align*}
	U = \text{diag}(\underbrace{\phi_1,\dots\phi_1,}_{\substack{l_1}} \underbrace{\phi_2,...,\phi_2,}_{\substack{l_2}}\dots, \underbrace{\phi_k,...,\phi_k}_{\substack{l_k}})
	\end{align*}
on $n$ qudits with $k$ distinct phases is initially decomposed into a product of a global phase and $k-1$ similar operator blocks:
\begin{align*}
	U&=\phi_1 \prod_{i=1}^{k-1} \diag(1,...1,\underbrace{\phi_{i+1}/\phi_i,\dots \phi_{i+1}/\phi_i}_{\substack{\sum_{j=i+1}^k l_j}})\text{.}
	\end{align*}

For the implementation of each block
 a phase gate 
	\begin{align*}
	P(\phi)&:=\ket{0}\bra{0}+\phi\ket{1}\bra{1}+\sum_{i=2}^{d-1} \alpha_i\ket{i}\bra{i}
	\end{align*}
(with arbitrary higher phases $\alpha_i$) assigns the desired phase $\phi=\phi_{i+1}/\phi_i$ to an ancillary target qudit if and only if its initial state $\ket{0}$ was changed to $\ket{1}$ beforehand by a so-called \emph{cascaded entangler} $\CINC(l)$ which checks if the original $n$-qudit register is in a computational basis state $\ket{j}$ with $j\ge d^n - l$:
\begin{align*}
 	\CINC(l)\ket{j}\ket{t}
 	:=\begin{cases}
 	\ket{j}\ket{t+1}, & \text{if }j \ge d^n - l\\
 	\ket{j}\ket{t}, & \text{if }j < d^n - l\text{.}
 	\end{cases}
 	\end{align*}
The operator $\diag(1,\dots,1,\phi,\dots\phi)$ with the last $l$ diagonal entries having value $\phi$ can hence be realized on a $n+1$ qudit system by
 	\begin{align*}
 	 V(\phi,l)=\CINC(l)^\dagger(\mathbb{I}_n \otimes P(\phi) ) \CINC(l)\text{.}
 	 \end{align*}

In the case of qubits a cascaded entangler $\CINC(l)$ is its own inverse. For $d>2$ one can realize the inverse by $\CINC(l)^\dagger=(\mathbb{I}_n\otimes M)\CINC(l)(\mathbb{I}_n\otimes M)$ with the single-qudit multiplication gate $M\ket{t}=\ket{- t}$.

The following circuit shows the decomposition of $V(\phi,l)$ for $d=2$, $n=6$ and $l=2$:
	
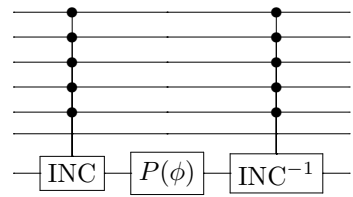
\begin{align*}
	\Qcircuit @C=1em @R=.7em {
	& \ctrl{6} & \qw & \ctrl{6} & \qw \\
	& \ctrl{5} & \qw & \ctrl{5} & \qw \\
	& \ctrl{4} & \qw & \ctrl{4} & \qw \\
	& \ctrl{3} & \qw & \ctrl{3} & \qw \\
	& \ctrl{2} & \qw & \ctrl{2} & \qw \\
	& \qw & \qw & \qw & \qw \\
	& \gate{\INC} & \gate{P(\phi)} & \gate{\INC^{-1}} & \qw 
	}
	\end{align*}
In this paper the most significant dit always occurs on the left of a written string and on the top of a drawn quantum circuit.

In the special case above the cascaded entangler corresponds to a single multi-controlled INC-gate (NOT-gate). This is due to the special choice of $l$ and not in general true. In the next section we will show how to decompose a general cascaded entangler into several multi-controlled INC-gates.

\section{Decomposition of cascaded entanglers}

\subsection{Multi-controlled $\INC$-operations}

The aim here is to decompose a cascaded entangler $\CINC(l)$ into elementary single-qudit gates and controlled $\INC$-operations with maximally 2 control levels. In a first step the cascaded entangler is decomposed into multi-controlled $\INC^{\pm1}$-operations that increase / decrease the ancillary target qudit iff the computational basis state $\ket{j}=\ket{j_1 \dots j_n}$ in the original $n$-qudit register corresponds in the first $m$ qudits to a specific state $\ket{b}=\ket{b_1\dots b_m}$, $b_i\in\{0,1,\dots, d-1\}$. These operations are denoted by
	\begin{align*}
	&\land^{n[b]}_m(\INC^{\pm 1})\ket{j}\ket{t}
	:=\begin{cases}
	\ket{j}\ket{t\pm 1} & \text{if } b=j_1\dots j_m\\
	\ket{j}\ket{t} & \text{otherwise.}
	\end{cases}
	\end{align*}

Notice that it is easy to change a $\land^{n[b]}_m(\INC^{-1})$-\ into a $\land^{n[b]}_m(\INC)$-operation by padding it by multiplication gates $M$ on the target qubit. By padding suitable addition gates $\INC_k\ket{t}=\ket{t+ k}$ on each control level it is furthermore possible to replace any $\land^{n[b]}_m(\INC)$-operator by a $\land^{n[11\dots 1]}_m(\INC)$-operator.
Afterwards the operator with $m>2$ control levels can be decomposed into linearly many doubly-controlled $\INC$-gates denoted by $\land_2$. The concrete parameters of the linear scaling depend on the chosen method. Here we exemplify the decomposition of a multi-controlled $\INC$-operation with $m=4$ control levels according to a generalization of the method of \cite{barenco} using $m-2$ ancilla qudits initialized to $\ket{0}$:
	
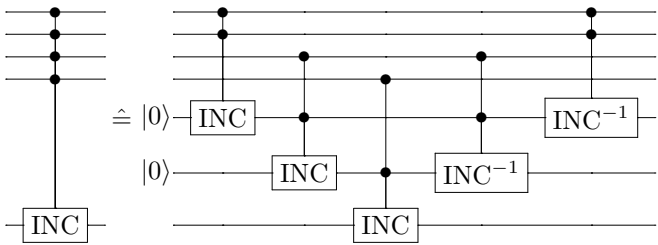
\begin{align*}
	\Qcircuit @C=0.65em @R=.6em {
	&\ctrl{6}		&\qw&&		&&&\ctrl{4}					&\qw&\qw&\qw						&\ctrl{4}&\qw  \\
	&\ctrl{5}		&\qw&&		&&&\ctrl{3}					&\qw&\qw&\qw						&\ctrl{3}&\qw \\
	&\ctrl{4}		&\qw&&		&&&\qw					&\ctrl{3}&\qw&\ctrl{3}				&\qw&\qw \\
	&\ctrl{3}		&\qw&&		&&&\qw					&\qw&\ctrl{3}&\qw					&\qw&\qw \\
	&			&&\hat{=}&		&{\ket{0}}&&\gate{\INC}  	&\ctrl{1}&\qw&\ctrl{1}				&\gate{\INC^{-1}}&\qw\\
	&			&&&			&{\ket{0}}&&\qw 		&\gate{\INC}&\ctrl{1}&\gate{\INC^{-1}}	&\qw&\qw \\
	&\gate{\INC}	&\qw&&		&&&\qw 					&\qw&\gate{\INC}&\qw 				&\qw&\qw 
	}
	\end{align*}

This specific decomposition method needs $2m-3$ $\land_2$-gates. The task of the last $m-2$ gates is to set the ancilla qudits back to $\ket{0}$. We could do without them, but most of them would cancel anyway with the $\land_2$-gates belonging to the next $\land_n^{m[b]}(\INC)$-operation and, additionally, they allow us to reuse the ancilla qudits and hence to keep the number of overall ancilla qudits small.

Since it is clear how to decompose $\land_m^{n[b]}(\INC^{\pm 1})$-operations into $\land_2$-gates, it only remains to study the decomposition of a cascaded entangler $\CINC(l)$ into $\land_m^{n[b]}(\INC^{\pm 1})$-operations.

\subsection{Results from classical logic synthesis}

In the qubit case the remaining decomposition task has a classical analogue in the $\{\land, \lor, \neg\}$-synthesis of the Boolean function $\phi: \{0,1\}^n \rightarrow \{0,1\}$ corresponding to the $\CINC(l)$-operation with the ancilla qubit interpretated as output. Let $\{s_i\}_{1\le i \le l}=\{x \in \{0,1\}^n | \phi(x)=1\}$ be the set of inputs satisfying 
$\phi$. A standard procedure in classical logic synthesis is to realize the circuit by the disjunctive normal form or sum-of-product form \cite[4.3]{hachtel}
\begin{align*}
\phi = \bigvee_{i=1}^l \bigwedge_{j=1}^n \neg^{s_{i,j} \oplus 1} x_j\text{.}
\end{align*}

This general method obviously scales with $\mathcal{O}(n2^n)$, since $l\in\mathcal{O}(2^n)$ in the worst case. Many heuristic optimisation methods are known making use of Karnaugh maps, BDDs, prime implicants and more \cite[II]{hachtel} resulting in practical algorithms such as ESPRESSO. However, none of these methods is actually capable of avoiding the exponential scaling in the general case. This is not surprising since  synthesizing the optimal circuit for an arbitrary Boolean formula is NP-hard. 


Of course Boolean functions corresponding to $\CINC(l)$-operations have a particular structure. They fall for example into the class of threshold functions obeying $\phi(x_1 \dots x_n) = 1$ iff $\sum_{i=1}^n w_i x_i$ greater than some threshold. For these functions \cite{thresholdFunctions} demonstrated a polynomially sized synthesis method and a scaling of $\mathcal{O}(n^2)$ in particular examples. 

This coincides with the scaling of $\mathcal{O}(dn^2)$ obeyed by the decomposition method for cascaded entanglers presented in the next section. The method is an expansion of \cite{welch} to qudit systems but is also directly formulated with another number-theoretical degree of freedom -- the choice of a signed base-d expansion for the natural numbers.







\subsection{Decomposition of a cascaded entangler into multi-controlled $\INC$-operations}

With the generalization
        \begin{align*}
 	\CINC(p,q)\ket{j}\ket{t}
 	:=\begin{cases}
 	\ket{j}\ket{t+1} & \text{if }p\le j <q\\
 	\ket{j}\ket{t-1} & \text{if }q\le j <p\\
 	\ket{j}\ket{t} & \text{otherwise}
 	\end{cases}
 	\end{align*}
one can easily find the trivial decomposition of a cascaded entangler corresponding to the classical disjunctive normal form, namely
\begin{align*}
\CINC(l)=\prod_{i=d^n-l}^{d^n-1} \CINC(i,i+1),
\end{align*}
where each $\CINC(i,i+1)$ already equals a desired multi-controlled $\INC$-operation $\land_n^{n[i]}(\INC)$. However, the number of $l$ multi-controlled $\INC$-gates each with $n$ control levels can be significantly reduced in many cases, e.g.\ consider the operation  $\CINC(2^{n-1})$ on a qubit system. In this example the cascaded entangler corresponds already to a single controlled $\INC$-gate with the first qubit as only control level. This feature, based on the structure of the binary representation of the number $l$, is exploited in the decomposition method by \cite{welch} which we generalize in this section to qudit systems. We start with two helpful lemmata concerning the operator $\CINC(p,q)$:

\begin{lem}\label{lem:split}
$\CINC(p,q)=\CINC(p,r)\cdot\CINC(r,q)$ for any $p,q,r\in\{0,1,\dots,d^n\}$.
\end{lem}
\begin{proof}
The different cases depending on the order relation of $p$, $q$ and $r$ can all be directly verified from the definition. Notice that $\CINC(p,p)=\mathbb{I}_{n+1}$ and $\CINC(p,q)=\CINC(q,p)^{-1}$.
\end{proof}

\begin{lem}\label{lem:multiControlled}
Suppose $p=bd^m$, $b\in\mathbb{N}_0$, $q=p+d^m$ and $p,q\in\left[0,d^n\right]$. Then $\CINC(p,q)=\land_{n-m}^{n[b]}(\INC)$.
\end{lem}
\begin{proof}
Since $q=(b+1)d^m\le d^n$, it holds $b<d^{n-m}$.
Let $b_1 b_2\dots b_{n-m}$ the $d$-ary representation of $b$. Then the $d$-ary representations of $p$ and $q-1$ turn out to be
\begin{align*}
	p=bd^m&= b_1 b_2 \dots b_{n-m} \underbrace{0\dots 0}_{\substack{m\text{ times}}}\\
	q-1=p+d^m -1&= b_1 b_2 \dots b_{n-m} \underbrace{(d-1) \dots (d-1)}_{\substack{m\text{ times}}}\text{.}
\end{align*}

$\CINC(p,q)$ increases the ancillary target qubit iff the original $n$-qudit register is found in a computational state $\ket{j}$ with $p\le j \le q-1$. According to the above $d$-ary representations this is exactly the case when the first $n-m$ qudits are in the state $\ket{b}$. Hence $\CINC(p,q)$ corresponds to a multi-controlled $\INC$-gate conditioned on the first $n-m$ qudits being in the state $\ket{b}$. This is directly the definition of $\land_{n-m}^{n[b]}(\INC)$.
\end{proof}

\begin{cor}\label{cor}
Suppose $p=bd^m$, $b\in\mathbb{N}_0$, $q=p-d^m$ and $p,q\in\left[0,d^n\right]$. Then $\CINC(p,q)=\land_{n-m}^{n[b-1]}(\INC^{-1})$.
\end{cor}
\begin{proof}
It holds $q=b' d^m$ with $b' = b-1\in\mathbb{N}_0$. Exchanging the roles of $p$ and $q$ in lemma \ref{lem:multiControlled} leads to $\CINC(p,q)=\CINC(q,p)^{-1}= \land_{n-m}^{n[b-1]}(\INC^{-1})$.
\end{proof}

The authors of \cite{welch} originally formulated their decomposition method for a cascaded entangler $\CINC(l)$ in a qubit system based on the binary representation of the parameter $l$. Later they adapted the method for a signed bit binary expansion of $l$. Here we don't just formulate the method for qudit systems but also directly for any \emph{signed base-$d$ expansion} of $l$. Such an expansion has the form $l=\sum_{i=1}^h s_i d^{m_i}$ with $0 \le m_1 \le m_2 \le \dots m_h\in\mathbb{N}_0$, $s_i=\pm 1$ and $d^n\ge\sum_{i=r}^h s_i d^{m_i}>0$ for all $1\le r \le h$. We require the bounds for the partial sums because they guarantee proper parameters for the $\CINC(p,q)$-gates used in the decomposition method.

\begin{thm}\label{thm}
Let $\sum_{i=1}^h s_i d^{m_i}$ be a signed base-$d$ expansion of $l$. Then $\CINC(l)=\prod_{i=1}^h \land_{n-m_i}^{n[b_i]}(\INC^{s_i})$.
\end{thm}
\begin{proof}
Define $p_i:=d^n-\sum\nolimits_{r=i}^h s_r d^{m_r}$ for all $i\in\{1,2,\dots, h\}$ and $p_{h+1}:= d^n$. It obviously holds that $p_1=d^n-l$ and $p_i\in\left[0,d^n\right]$ for all $i\in\{1,2,\dots,h+1\}$. According to lemma \ref{lem:split} we can decompose
	\begin{align*}
	\CINC(l)=\CINC(p_1,p_{h+1})=\prod_{i=1}^h \CINC(p_i,p_{i+1})\text{.}
	\end{align*}

Because $p_i$ is divisible by $d^{m_i}$ for all $1\le i \le h$, we can write $p_i=b'_i d^{m_i}$ with $b'_i\in\mathbb{N}_0$ and  $p_{i+1}=p_i+s_i d^{m_i}$.
Since the requirements of lemma \ref{lem:multiControlled} and corollary \ref{cor} are fulfilled it follows
\begin{align*}
	\CINC(p_i,p_{i+1})=\land_{n-m_i}^{n[b_i]}(\INC^{s_i})
	\end{align*} 
with $b_i=b'_i$ in the case $s_i=1$ and $b_i=b'_i-1$ in the case $s_i=-1$. This completes the proof.
\end{proof}

With the previous theorem we completed the decomposition of an arbitrary diagonal operator into $\land_1$-, $\land_2$- and basic single-qudit gates. If we assume the setting of few distinct phases $k\in\mathcal{O}(\poly(n))$ and apply the previous theorem based on the standard $d$-ary expansion, the decomposition is even efficient. However, for many cascaded entanglers $\CINC(l)$ there exists an alternative signed base-$d$ expansion of $l$ leading to a decomposition into multi-controlled $\INC$-gates with a significantly smaller number of overall control levels and hence of required $\land_1$- and $\land_2$-gates due to the linear dependence. We close the section by demonstrating this in two examples:

\begin{ex}\label{ex:1}
$d=2$, $n=3$, $l=7$.
\vspace{0.5ex}

\hspace{0.1em}Standard: $l=2^0+2^1+2^2$\hspace{1.9em}Signed: $l=-2^0+2^3$
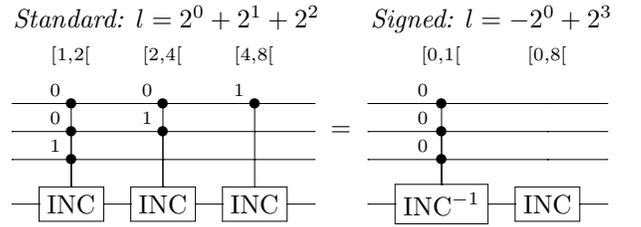
\vspace{0.5ex}
	 \begin{align*}
	 \Qcircuit @C=1em @R=.8em {
	 &\ustick{_{\left[1,2\right[}}&\ustick{_{\left[2,4\right[}}&\ustick{_{\left[4,8\right[}}&		&&		&\ustick{_{\left[0,1\right[}}&\ustick{_{\left[0,8\right[}}&\\
	 &\ctrl{3}_{\quad 0}		&\ctrl{3}_{\;\qquad 0}	&\ctrl{3}_{\;\qquad 1}	&\qw		&&		&\ctrl{3}_{\ \quad 0}		&\qw					&\qw\\
	 &\ctrl{2}_{\quad 0}		&\ctrl{2}_{\;\qquad 1}	&\qw					&\qw		&=&		&\ctrl{2}_{\ \quad 0}		&\qw					&\qw\\
	 &\ctrl{1}_{\quad 1}		&\qw 				&\qw					&\qw		&&		&\ctrl{1}_{\ \quad 0}		&\qw 				&\qw\\
	 &\gate{\INC}			&\gate{\INC}			&\gate{\INC}			&\qw		&&		&\gate{\INC^{-1}}		&\gate{\INC}			&\qw
	 } 
	 \end{align*}
\end{ex}

\begin{ex}\label{ex:2}
$d=5$, $n=4$, $l=14$.
\vspace{0.5ex}

\hspace{0.2em}Standard: $ l=4\cdot 5^0 + 2\cdot 5^1$
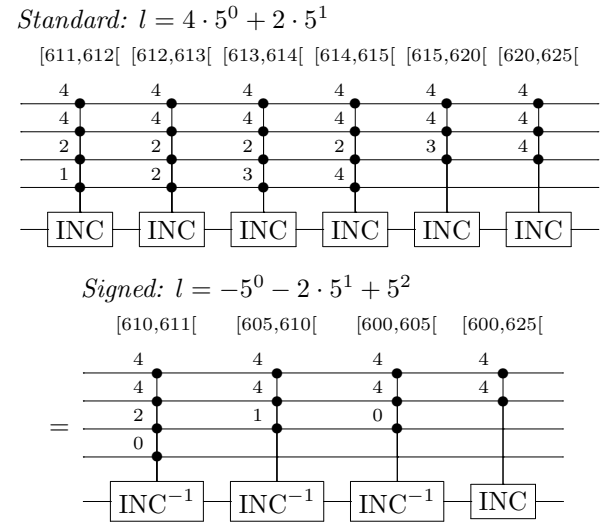
\hspace{1.9em}
\vspace{0.5ex}
	 \begin{align*}
	 \Qcircuit @C=1em @R=.8em {
&\ustick{_{\left[611,612\right[}}&\ustick{_{\left[612,613\right[}}&\ustick{_{\left[613,614\right[}}&\ustick{_{\left[614,615\right[}}&\ustick{_{\left[615,620\right[}}&\ustick{_{\left[620,625\right[}}&\\	
	 &\ctrl{4}_{\quad 4}	&\ctrl{4}_{\;\qquad 4}	&\ctrl{3}_{\;\qquad 4}	&\ctrl{4}_{\;\qquad 4}	&\ctrl{4}_{\;\qquad 4}	&\ctrl{3}_{\;\qquad 4}	&\qw\\
	 &\ctrl{3}_{\quad 4}	&\ctrl{3}_{\;\qquad 4}	&\ctrl{3}_{\;\qquad 4}	&\ctrl{3}_{\;\qquad 4}	&\ctrl{3}_{\;\qquad 4}	&\ctrl{3}_{\;\qquad 4}	&\qw\\
	 &\ctrl{2}_{\quad 2}	&\ctrl{2}_{\;\qquad 2}	&\ctrl{2}_{\;\qquad 2}	&\ctrl{2}_{\;\qquad 2}	&\ctrl{2}_{\;\qquad 3}	&\ctrl{2}_{\;\qquad 4}	&\qw\\
	 &\ctrl{1}_{\quad 1}	&\ctrl{1}_{\;\qquad 2}	&\ctrl{1}_{\;\qquad 3}	&\ctrl{1}_{\;\qquad 4}	&\qw  				&\qw					&\qw\\
	 &\gate{\INC}	&\gate{\INC}				&\gate{\INC}			&\gate{\INC}			&\gate{\INC}			&\gate{\INC}			&\qw
	 } 
	 \end{align*}

\vspace{0.5ex}

\hspace{2.6em}Signed: $ l=-5^0-2\cdot 5^1+5^2$
\vspace{0.5ex}
	 \begin{align*}
	 \Qcircuit @C=1em @R=.8em {
	 &&\ustick{_{\left[610,611\right[}}&\ustick{_{\left[605,610\right[}}&\ustick{_{\left[600,605\right[}}&\ustick{_{\left[600,625\right[}}&\\	 
	 &&\ctrl{4}_{\ \quad 4} 		&\ctrl{4}_{\quad\qquad 4}	&\ctrl{3}_{\quad\qquad 4}	&\ctrl{4}_{\ \,\qquad 4}	&\qw\\
	 &&\ctrl{3}_{\ \quad 4}		&\ctrl{3}_{\quad\qquad 4}	&\ctrl{3}_{\quad\qquad 4}	&\ctrl{3}_{\ \,\qquad 4}	&\qw\\
	 &\lstick{=}&\ctrl{2}_{\ \quad 2}	&\ctrl{2}_{\quad\qquad 1}	&\ctrl{2}_{\quad\qquad 0}	&\qw					&\qw\\
	 &&\ctrl{1}_{\ \quad 0}		&\qw					&\qw					&\qw					&\qw\\
	 &&\gate{\INC^{-1}}			&\gate{\INC^{-1}}		&\gate{\INC^{-1}}		&\gate{\INC}			&\qw
	 } 
	 \end{align*}
\end{ex}

\section{Conclusion}

\subsection{Summary}



We generalized an efficient decomposition algorithm for diagonal phase-sparse unitaries presented in \cite{welch} to qudit systems. 
While this generalization is interesting from a number-theoretical point of view, it might also be advantageous for practical implementations since it allows decompositions into multi-controlled $\INC$-gates with less control levels compared to the qubit case (though higher-dimensional).

An advantage of the presented algorithm over other decomposition methods is its consideration of the phase context of the unitary which leads to the small number $k$ of required single-qudit phase gates. Hence the decomposition of an $n$-qubit operator with $k=2$ distinct phases 
only requires two single-qubit gates while previous methods decompose into $\Omega(2^n)$ phase gates as it was already pointed out by  \cite{welch}. 
The number of required phase gates is of particular interest since they form the accuracy dependent part in this decomposition (in contrast to the exactly decomposable cascaded entanglers) in the case the single-qudit operations are further decomposed into some approximating, eventually finite set.

In the worst case -- e.g. a signed base-d expansion $\sum_{j=0}^{n-1} (d-1) d^j$ -- a cascaded entangler is decomposed into multi-controlled $\INC^{\pm 1}$ with a total number of  $(d-1)\sum_{j=1}^n j=\mathcal{O}(dn^2)$ control levels. This results in a total number of $\mathcal{O}(dn^2k)$ $\land_1$- and $\land_2$-gates in the decomposition as well as $\mathcal{O}(dn^2k)$ many single-qubit operations. 


One accomplishment of this paper is the formulation of the decomposition algorithm based on arbitrary signed base-$d$ extensions of natural numbers which may allow a significant reduction of required gates over the standard $d$-ary extension as seen in examples \ref{ex:1} and \ref{ex:2}.

\subsection{Outlook and open questions}

It was verified by brute force that the signed base-$d$ expansions in examples \ref{ex:1} and \ref{ex:2} are indeed those that lead to the minimum of overall control levels as well as the simultaneous minimum of required $\land_1$- and $\land_2$-gates according to the presented further $\land_2$-decomposition scheme.
Unfortunately there is no efficient algorithm known for the computation of an optimal signed base-$d$ expansion. This is an open question even for the qubit case \cite{welch}. 
In the higher-dimensional case the tradeoff between the summands (multi-controlled $\INC$-gates) and their exponents (control levels) depends moreover on the qudit dimension $d$ and hence turns into an even more complicated multi-parameter optimization problem.

Of course one can at least improve the performance over the standard $d$-ary representation by considering other efficiently computable signed base-$d$ expansion schemes in comparision. In this spirit \cite{welch} propose a specific recursive algorithm and numerically confirm that it outperforms the standard binary expansion for most natural numbers. Such an algorithm is easy to adapt for the qudit case. 

It seems plausible that the way over multi-controlled $\INC$-operations leads to the minimal number of $\land_1$- and $\land_2$-gates required for the decomposition of a cascaded entangler (taking cancelling effects into consideration).
Thinking about the optimal decomposition of cascaded entanglers into multi-controlled $\INC$-operations it seems moreover intuitively reasonable to consider only those schemes
which directly correspond to a signed base-$d$ expansion of the represented number. If this intuition should be confirmed, the question of \cite{welch} about the complexity of cascaded entanglers is equivalent to the question of the optimal signed base-$d$ expansion of natural numbers.

Beyond the reduction of the problem to cascaded entanglers it remains to study other decomposition methods for diagonal unitaries under the aspect of phase-spareness in order to improve the scaling of $\mathcal{O}(dn^2k)$. A compelling but non-trivial candidate for this is the best known qubit algorithm \cite{bullockOptimal} with a scaling of $\mathcal{O}(2^n)$ without phase-context consideration.




\begin{acknowledgments}
This  work  was  supported  by  the  ERC grants QFTCMPS, and SIQS, and through the DFG by the cluster of excellence EXC 201 Quantum FQ Engineering and Space-Time Research, and the Research Training Group 1991. We thank Tobias J. Osborne for the idea and the supervision of the bachelor thesis from which the presented result originates.
\end{acknowledgments}

\bibliographystyle{plain} 

\end{document}